\documentclass[conference, final]{IEEEtran} 
\usepackage{amsmath, amsthm, amssymb}
\usepackage{epsfig}
\usepackage{latexsym}
\usepackage{graphicx}
\usepackage{cite}
\usepackage{subfigure}

\newtheorem{theorem}{Theorem}
\newtheorem{lemma}{Lemma}

\newtheorem{corollary}{Corollary}  
\newtheorem{remark}{Remark}

\title{Analytic Solutions to a Marcum $Q{-}$Function-Based Integral and Application in Energy Detection of Unknown Signals over Multipath Fading Channels}

\author{
Paschalis~C.~Sofotasios$^{1,2}$,~Mikko~Valkama$^{1}$,~Theodoros~A.~Tsiftsis$^{2,3}$,~Yury~A.~Brychkov$^{4}$, \\Steven~Freear$^{5}$~and~George~K.~Karagiannidis$^{2,6}$
\\\\
\begin{normalsize} 
$^{1}$Department of Electronics and Communications Engineering, Tampere University of Technology, 33101 Tampere, Finland. 
\end{normalsize}
\\
\begin{normalsize} 
e-mail: $\rm \left\lbrace paschalis.sofotasios; mikko.e.valkama \right\rbrace@\rm tut.fi$ 
\end{normalsize} 
\\
\begin{normalsize} 
$^{2}$Department of Electrical and Computer Engineering, Aristotle University of Thessaloniki, 54124 Thessaloniki, Greece.
\end{normalsize}
\\
\begin{normalsize} 
e-mail:  $ \rm geokarag@\rm ieee.org$
\end{normalsize} 
 \\
\begin{normalsize} 
$^{3}$Department of Electrical Engineering, Technological Educational Institute of Central Greece, 35100 Lamia, Greece. 
\end{normalsize}
\\
\begin{normalsize} 
e-mail: $\rm  tsiftsis@\rm teilam.gr$
\end{normalsize} 
\\
\begin{normalsize} 
$^{4}$The Dorodnicyn Computing Center of the Russian Academy of Sciences, 119333 Moscow, Russia. 
\end{normalsize} 
\\
\begin{normalsize} 
e-mail: $\rm brychkov@\rm ccas.ru$
\end{normalsize} 
\\
\begin{normalsize} 
$^{5}$School of Electronic and Electrical Engineering, University of Leeds, LS2 9JT, Leeds, United Kingdom. 
\end{normalsize}
\\
\begin{normalsize} 
e-mail: $\rm  s.freear@\rm leeds.ac.uk$
\end{normalsize} 
\\
\begin{normalsize} 
$^{6}$Department of Electrical and Computer Engineering, Khalifa University, PO Box 127788, Abu Dhabi, UAE. 
\end{normalsize} 
 }

\begin{document}
\maketitle

\begin{abstract}
This work presents  analytic solutions for a useful integral in wireless communications, which involves the Marcum $Q{-}$function in combination with an  exponential function and arbitrary power terms. The derived expressions have a rather simple algebraic representation which renders them convenient both analytically and computationally. Furthermore, they can be useful in wireless communications and particularly in the context of cognitive radio communications and radar systems, where this integral is often encountered. To this end, we  derive novel expressions for the probability of detection in energy detection based spectrum sensing over $\eta{-}\mu$ fading channels.  These expressions are given in closed-form and are subsequently employed in analyzing the effects of generalised multipath fading conditions in cognitive radio systems. As expected, it is shown that the detector is highly dependent upon the severity of fading  conditions as even slight variation of the fading parameters affect  the corresponding  performance significantly. 
\end{abstract}

\section{Introduction}

The generalized Marcum $Q{-}$function, $Q_{m}(a,b)$, is a vital special function in wireless communication theory. It was proposed several decades ago and  has appeared extensively in various analyses in the context of stochastic processes in probability theory,  single- and multi-channel based communications over fading channels, information-theoretic analysis of multiple-input-multple-output (MIMO) systems, cognitive radio and radar systems, among others \cite{J:Marcum_2, J:Swerling, J:Nuttall_2, J:Simon, B:Alouini, J:Karagiannidis, J:Brychkov2012, B:Vasilis_PhD, J:Brychkov2013}, and the references therein. Its use has also  led to the derivation of numerous tractable analytic expressions, while its computational realization is rather straightforward since it is included as a built-in function in the most popular software packages \cite{Math_1, Math_4, Math_5, Math_6, Add_0, Math_9}.

However, it is widely known that the derivation of tractable analytic expressions in natural sciences and engineering can be rather laborious and cumbersome, if not impossible,  as integrals that involve combinations of elementary and special functions are often required to be evaluated analytically \cite{Math_11, Math_3, Math_7, Math_8, Math_10, Add}, and the references therein. This is also the case when the   Marcum $Q{-}$function is involved in integrands along with exponential and arbitrary power terms. A general form of such an integral is  the following:  
\begin{equation} \label{Integral}
\mathcal{I}_{a,b}(k, m, p) =   \int_{0}^{\infty} x^{2k-1} Q_{m}(ax, b) e^{-px^{2}} dx
\end{equation}
which can be equivalently expressed as
\begin{equation} \label{second}
\mathcal{I}_{a,b}(k, m, p) = \frac{1}{2} \int_{0}^{\infty} x^{k-1} Q_{m}(a\sqrt{x}, b) e^{-px} dx.  
\end{equation} 
The  integrals in \eqref{Integral} and \eqref{second} are  encountered in various applications relating to wireless communications, such as in the analysis of multichannel diversity systems with non-coherent and differentially coherent detection and  in  sensing of unknown signals in the context of cognitive radio and radar systems  \cite{J:Alouini, B:Bargava, J:Urkowitz, C:Kostylev, C:Tellambura_2, J:Ghasemi, C:Attapattu_2, J:Janti, C:Sofotasios, B:Sofotasios, J:Herath,  C:Beaulieu, Add_1, J:Sofotasios, Add_2, Add_3, Add_4, Add_5, Add_6} and the references therein. Based on this, a recursive formula restricted to only integer values of $k$ and $m$ was reported in \cite{J:Nuttall_2} while an infinite series representation for the case that $k$ is arbitrary and $m$ is positive integer was recently reported in \cite{J:Ran}.

Nevertheless, these expressions are neither generic, nor account for the case that $m$ is an arbitrary real.  Motivated by this, this work is devoted to the derivation of analytic expressions for $\mathcal{I}_{a,b}(k, m, p) $ which can be useful in applications relating to the wide field of digital communications. To this end, novel analytic expressions are derived for the probability of detection  of energy detection based spectrum sensing over generalized $\eta{-}\mu$ fading channels. The derived expressions are given in closed-form and are utilized in analyzing the performance of the detector in various fading conditions.

The remainder of this paper is organized as follows: Sec. II provides the derivation of two analytic expressions for $\mathcal{I}_{a,b}(k, m, p) $. Sec. III is devoted to the application of the offered results in the analytic performance evaluation of energy detection based spectrum sensing over $\eta{-}\mu$ fading channels for various severity scenarios. The corresponding numerical results are given in Sec. IV along with useful discussions, while closing remarks are provided in Sec. V.

\section{Analytic Solutions to $\mathcal{I}_{a,b}(k, m, p)$ Integrals}

\subsection{A Closed-Form Expression for Arbitrary Integer Values of $k$ and Arbitrary Real Values of $m$}

As already mentioned, no analytic expressions for \eqref{Integral} and \eqref{second} for the case of arbitrary real values of $m$ have been reported in the open scientific and technical literature.

\begin{theorem}
For  $a, b, m, p \in \mathbb{R^{+}}$ and $k \in \mathbb{N}$, the following closed-form expressions is valid, 
\begin{equation} \label{T_1}
\mathcal{I}_{a,b}(k, m, p) = \frac{\Gamma(k) \Gamma \left( m, \frac{b^{2}}{2} \right)}{2 p^{k} \Gamma(m)}  \qquad \qquad \qquad \qquad \qquad \qquad \qquad \qquad \qquad 
 \end{equation}
 \begin{equation*}
\qquad \, \qquad + \sum_{ l = 0}^{k -1} \frac{ a^{2} b^{2m}  \Gamma(k) \,_{1}F_{1} \left( l + 1, m + 1, \frac{a^{2}b^{2}}{ 2a^{2} + 4p } \right)  }{m!   p^{k-l} 2^{m - l+1} \left( a^{2} + 2p \right) ^{l+1} e^{\frac{b^{2}}{2}}} 
\end{equation*} 
where $\Gamma(x)$ and $\Gamma(x,a)$ denote the gamma function and upper incomplete gamma function, respectively, $x! \triangleq \Gamma(x-1)$ is the increasing factorial and $\, _{1}F_{1}(x, y, z)$ is the Kummer confluent hypergeometric function \cite{B:Prudnikov, B:Tables_1, B:Tables_2}.
\end{theorem}

\begin{proof}
By integrating \eqref{second} by parts  one obtains, 
\begin{equation}\label{T_2a}
\begin{split}
\mathcal{I} _{a,b}(k, m, p) &= \overbrace{  \lim_{c \to \infty} \left[ \frac{Q_{m}(a \sqrt{x},b)}{2}   \int x^{k-1} e^{-px} dx \right]_{0}^{c} }^{\mathcal{T}} \\
& - \frac{1}{2} \int_{0}^{\infty}  \left[ \int \frac{x^{k-1}}{ e^{px}} dx \right]\, \frac{d}{dx} Q_{m}(a \sqrt{x},b) dx
\end{split}
\end{equation}
where $c$ is a non-negative finite real. By recalling that the lower incomplete gamma function is given by $\gamma(a,x) \triangleq \Gamma(a) - \Gamma(a,x)$, it readily follows that $\int x^{a-1} {\rm exp}(-x)dx = \gamma(a,x) = - \Gamma(a,x)$. Upon substituting in \eqref{T_2a} one obtains,
\begin{equation} \label{T_2}
\begin{split}
\mathcal{T} &=   \lim_{c \to \infty} \left[ \frac{Q_{m}(a \sqrt{x},b) \gamma(k, px)}{2p^{k}}    \right]_{0}^{c} \\
& =  \frac{Q_{m}(0,b) \Gamma(k, 0)}{2p^{k}}  -  \lim_{c \to \infty}  \frac{Q_{m}(a \sqrt{c},b) \Gamma(k, pc)}{2p^{k}}.    
\end{split}
\end{equation}
With the aid of the identities for $Q_{m}(a,b)$ and $\Gamma(a,x)$ functions in \cite{J:Nuttall_1, J:Nuttall_2} and \cite{B:Prudnikov} as well as expressing  $d Q_{m}(a,b){/}da$ according to \cite[eq. (10)]{J:Nuttall_2},  it follows that
\begin{equation} \label{T_3}
\begin{split}
\mathcal{I} _{a,b}(k, m, p) &= \sum_{l = 0}^{m-1} \frac{\Gamma(k) b^{2l}}{l! p^{k} 2^{l}} e^{- \frac{b^{2}}{2}} \\
& + \frac{b^{m}}{p^{k} e^{\frac{b^{2}}{2}}} \int_{0}^{\infty} \frac{  \Gamma \left(k, \frac{p y^{2}}{a^{2}} \right)  e^{ - \frac{y^{2} }{2}} I_{m}(by)  }{ y^{m-1}  } dy. 
\end{split}
\end{equation}
By making the necessary variable transformation in \cite[II. 7 - pp. 726]{B:Tables_1} and substituting in \eqref{T_3} yields
\begin{equation} \label{T_4}
\mathcal{I} _{a,b}(k, m, p) =  \sum_{l = 0}^{m-1} \frac{\Gamma(k) b^{2l}}{l! p^{k} 2^{l}} e^{- \frac{b^{2}}{2}}  \qquad \qquad \qquad \qquad   \qquad  
\end{equation}  
\begin{equation*}
\, \, \qquad \qquad \quad  + \sum_{l = 0}^{k -1} \frac{b^{m} \Gamma(k) e^{ - \frac{b^{2}}{2}} }{l! p^{k - l} a^{2l}}  \int_{0}^{\infty} \frac{y^{2l - m + 1}  I_{m}(by)}{e^{\frac{y^{2}}{2} \left( 1 + \frac{2p}{a^{2}} \right) }} dy. 
\end{equation*}
Notably, the above integral can be expressed in closed-form with the aid of \cite[eq. (2.15.5.4)]{B:Tables_1}. To this effect, by performing the necessary change of variables and substituting in \eqref{T_4}, equation \eqref{T_1} is deduced, which completes the proof. 
\end{proof}
$ $ \\
It is noted that the algebraic representation of the derived solution is simpler than the recursive expression in \cite{J:Nuttall_2, J:Simon} and additionally,  the value of $m$ is subject to no restrictions.

\subsection{An Exact Infinite Series to $\mathcal{I}_{a,b}(k,m,p)$ for Arbitrary Reals}

It is recalled that no analytic expressions exist  for \eqref{Integral} and \eqref{second}   for arbitrary real values, i.e. unrestricted, of all involved parameters.

\begin{lemma}
For $a, b, k, m, p \in \mathbb{R^{+}}$, the following exact infinite series representation is valid for the integral in \eqref{Integral} and \eqref{second}, 
\begin{equation} \label{T_a}
\mathcal{I}_{a,b}(k, m, p) = \sum_{l = 0}^{\infty} \frac{a^{2l} 2^{k}  \Gamma(k + l) \Gamma \left(m + l, \frac{b^{2}}{2} \right) }{l!  \Gamma(m + l) \left( a^{2} + 2p \right)^{k + l}}.  
\end{equation}
\end{lemma} 

\begin{proof}
The $Q_{m}(a, b)$  function can be expressed in  infinite series according to  \cite[eq. (29)]{J:Karagiannidis}. Therefore, by performing the necessary change of variables it immediately follows that, 
\begin{equation} \label{Marcum_Series}
Q_{m}(a\sqrt{x},b) = e^{-\frac{x a^{2}}{2}} \sum_{l = 0}^{\infty} \frac{a^{2l} x^{l} \Gamma\left(l + m, \frac{b^{2}}{2} \right)}{l! 2^{l} \Gamma(l + m)}
\end{equation}
which upon substitution in \eqref{second} yields,  
\begin{equation} \label{T_b}
\mathcal{I}_{a,b}(k, m, p) = \sum_{l = 0}^{\infty}  \frac{a^{2l} \Gamma\left( m + l, \frac{b^{2}}{2} \right)}{l! 2^{l} \Gamma(m + l) }  \underbrace{ \int_{0}^{\infty}  \frac{  x^{k + l - 1}}{  e^{x \left(p + \frac{a^{2}}{2} \right)}}   dx}_{\mathcal{R}}. 
\end{equation}
The above integral can be expressed in terms of the $\Gamma(.)$ function in \cite[eq. (2.10.3.2)]{B:Tables_1}, yielding 
\begin{equation} \label{T_b_1}
\mathcal{R} = \frac{ 2^{k + l} \Gamma(k + l)}{a^{2} + 2p}.
\end{equation}
Evidently, by substituting \eqref{T_b_1} in \eqref{T_b}, equation \eqref{T_a} is deduced thus completing the proof.  
\end{proof}

The series in \eqref{T_a} is convergent and can provide acceptable accuracy when truncated after relatively few terms. However, deriving a closed-form expression for the truncation error is particularly advantageous in determining the corresponding truncation error   accurately and straightforwardly.

\begin{lemma}
For $a, b, k, m, p \in \mathbb{R^{+}}$, the following inequality can serve as a closed-form upper bound for the truncation error of  \eqref{T_a},  
\begin{equation} \label{T_c}
\epsilon_{t} \leq  \frac{\Gamma(k)}{p^{k}}  - \sum_{l = 0}^{n} \frac{a^{2l} 2^{k}  \Gamma(k + l) \Gamma \left(m + l, \frac{b^{2}}{2} \right) }{l!  \Gamma(m + l) \left( a^{2} + 2p \right)^{k + l}}.  
\end{equation}
\end{lemma}

\begin{proof}
The truncation error of \eqref{T_a} when this is truncated after $n$ terms is expressed as,
\begin{equation} \label{Tr_1}
\begin{split}
\epsilon_{t} &= \sum_{l = n + 1}^{\infty} \frac{a^{2l} 2^{k}  \Gamma(k + l) \Gamma \left(m + l, \frac{b^{2}}{2} \right) }{l!  \Gamma(m + l) \left( a^{2} + 2p \right)^{k + l}} \\
& = \sum_{l = 0}^{\infty} \frac{a^{2l} 2^{k}  \Gamma(k + l) \Gamma \left(m + l, \frac{b^{2}}{2} \right) }{l!  \Gamma(m + l) \left( a^{2} + 2p \right)^{k + l}} \\ 
& \quad - \sum_{l = 0}^{n} \frac{a^{2l} 2^{k}  \Gamma(k + l) \Gamma \left(m + l, \frac{b^{2}}{2} \right) }{l!  \Gamma(m + l) \left( a^{2} + 2p \right)^{k + l}}. 
\end{split}
\end{equation}
It is recalled  that the $\Gamma(a, x)$ function is monotonically decreasing  w.r.t.  $x$ and thus, $\Gamma(a,x) \leq \Gamma(a)$. To this effect, the upper incomplete gamma function in \eqref{Tr_1} can be bounded as,  
\begin{equation} \label{Tr_2}
\epsilon_{t} \leq  \sum_{l = 0}^{\infty} \frac{a^{2l} 2^{k}  \Gamma(k + l)  }{l!    \left( a^{2} + 2p \right)^{k + l}} - \sum_{l = 0}^{n} \frac{a^{2l} 2^{k}  \Gamma(k + l) \Gamma \left(m + l, \frac{b^{2}}{2} \right) }{l!  \Gamma(m + l) \left( a^{2} + 2p \right)^{k + l}}. 
\end{equation} 
By recalling  the Pochhammer symbol,  $(a)_{n} \triangleq \Gamma(a + n){/}\Gamma(a)$, it follows that  $\Gamma(k + l) = (k)_{l} \Gamma(k)$. Based on this,  the above infinite series can be expressed as follows, 
\begin{equation} \label{tr_series}
\sum_{l = 0}^{\infty} \frac{a^{2l} 2^{k}  \Gamma(k + l)  }{l!    \left( a^{2} + 2p \right)^{k + l}}  = \frac{\Gamma(k) 2^{k}  }{ \left( a^{2} + 2p \right)^{k}} \sum_{l = 0}^{\infty} \frac{a^{2l}  (k)_{l}    }{l!    \left( a^{2} + 2p \right)^{  l}  }.  
\end{equation}
The infinite series in the right-hand side of \eqref{tr_series} can be expressed in terms of the hypergeometric function, namely, 
\begin{equation} \label{tr_series_1}
\sum_{l = 0}^{\infty} \frac{a^{2l}  (k)_{l} (1)_{l}   }{l!    \left( a^{2} + 2p \right)^{  l} (1)_{l}} = \,_{1}F_{0}\left(k; ; \frac{a^{2}}{ a^{2} + 2p} \right) 
\end{equation}
 Based on \eqref{tr_series_1}, it immediately follows that
\begin{equation} \label{additional}
\,_{1}F_{0} \left(k;\, ; \frac{a^{2}}{ a^{2} + 2p} \right) = \frac{(a^{2} + 2p)^{k}}{2^{k} p^{k}}. 
\end{equation}
Therefore, by substituting in \eqref{additional} in \eqref{tr_series} one obtains, 
\begin{equation}  \label{additional_2}
\sum_{l = 0}^{\infty} \frac{a^{2l} 2^{k}  \Gamma(k + l)  }{l!    \left( a^{2} + 2p \right)^{k + l}}  = \frac{\Gamma(k)}{p^{k}}. 
\end{equation}
Evidently,  by substituting \eqref{additional_2} in  \eqref{Tr_2} yields \eqref{T_c} thus, completing the proof. 
\end{proof}

To the best of the Authors' knowledge, equations \eqref{T_1}, \eqref{T_a} and \eqref{T_c} have not been previously reported in the open technical literature.

\section{Applications in Cognitive Radio}

\subsection{Energy Detection Based Spectrum Sensing}

Cognitive radio (CR) is an emerging technology that allows opportunistic access of licensed frequency bands when they are not utilized. Given the increased spectrum scarcity along with the high demands for bandwidth resources, CR is anticipated to play a core role in the next generation of mobile communication systems, namely 5G. The most important part of CR technology is the accurate and robust sensing of vacant frequency bands and based on the respective decision the user will decide on whether it can establish communication or not. Therefore, spectrum sensing is the most critical  operation in CR systems with energy detection being regarded as the most simple and popular method \cite{B:Bargava}. In this context, the performance of energy detection based spectrum sensing over various fading conditions have been investigated in \cite{J:Urkowitz, C:Kostylev, C:Tellambura_2, J:Ghasemi, C:Attapattu_2, J:Janti, C:Sofotasios, B:Sofotasios, J:Herath, C:Beaulieu, J:Sofotasios} - and the reference therein.  

It is recalled that in narrowband energy detection, the received signal waveform follows a binary hypothesis that can be represented as \cite[eq. (1)]{C:Beaulieu},
\begin{equation} \label{Test_1} 
r(t) =
\begin{cases}
n(t) \,\,\,\, \qquad \,\,\,\, \qquad \,\,\, :H_{0} \\
hs(t) + n(t) \, \qquad \, \,:H_{1}
\end{cases}
\end{equation}
where $s(t)$, $h$ and $n(t)$ denote an unknown deterministic signal,  the amplitude of the channel coefficient and  an additive white Gaussian noise (AWGN) process, respectively. The samples of $n(t)$ are assumed to be zero-mean Gaussian random variables with variance $N_{0}W$ with $W$ and $N_{0}$ denoting the  single-sided signal bandwidth and a single-sided noise power spectral density, respectively  \cite{C:Beaulieu}. The hypotheses $H_{0}$ and $H_{1}$ refer to the cases that a signal is absent or present, respectively.  The received signal is subject to filtering, squaring and integration over the time interval $T$ which is expressed as \cite[eq. (2)]{J:Alouini}, $y \triangleq \frac{2}{N_{0}} \int_{0}^{T} \mid r(t)\mid ^{2} dt$. The output of the integrator corresponds to a measure of the energy of the received waveform and acts as a test statistic that determines whether the received energy measure corresponds only to the energy of noise ($H_{0}$) or to the energy of both the unknown deterministic signal and noise ($H_{1}$). By denoting the time bandwidth product as $u = TW$,  the test statistic typically follows the central chi-square distribution with $2u$ degrees of freedom under the $H_{0}$ hypothesis and the non central chi-square distribution with $2u$ degrees of freedom under the $H_{1}$ hypothesis \cite{J:Urkowitz}.  Based on this and by recalling that energy detection is largely affected by a predefined energy threshold, $\lambda$, the performance of the detector is characterized by the probability of false alarm, $P_{f}={\rm Pr}(y> \lambda \mid H_{0})$ and the probability of detection, $P_{d}= {\rm Pr}(y> \lambda \mid H_{1})$, namely  \cite{J:Alouini}, 
\begin{equation} \label{Pf_1} 
P_{f} = \frac{\Gamma \left(u,\frac{\lambda}{2}\right)}{\Gamma(u)}
\end{equation}
and
\begin{equation}\label{Pd_1} 
P_{d} = Q_{u}\left(\sqrt{2 \gamma},\sqrt{\lambda} \right)
\end{equation}
respectively.

\subsection{The $\eta{-}\mu$ Distribution}

The $\eta{-}\mu$ distribution is a generalized fading model that has been widely shown to provide adequate characterization of multipath fading in non-line-of-sight (NLOS) communications. It was reported in \cite{J:Yacoub} along with the $\kappa{-}\mu$ fading model which accounts for corresponding line-of-sight (LOS) communication scenarios. The $\eta{-}\mu$ fading model has been shown to be particularly flexible and it includes as special cases the well known  Hoyt, Nakagami${-}m$, Rayleigh and one-sided Gaussian distributions \cite{J:Yacoub}.  Its remarkable flexibility and usefulness were demonstrated clearly in \cite[Fig. 9]{J:Yacoub} along with the $\kappa{-}\mu$ fading model where it is clearly shown that the $\eta{-}\mu$ fading model is significantly more flexible than the more commonly adopted Nakagami${-}m$ and Rayleigh distributions.

 In terms of physical interpretation, the $\eta{-}\mu$ fading model is expressed by two physical parameters, $\eta$ and $\mu$ and it holds for two formats, namely \textit{Format-$1$} and \textit{Format-$2$}.  In the former, the $\eta$ parameter denotes the ratio of the powers between the multipath waves in the in-phase and quadrature components, whereas in the latter it denotes the correlation coefficient between the scattered wave in-phase and quadrature components of each cluster of multipath. Likewise, the $\mu$ parameter denotes - in both formats - the inverse of the normalised variance and relates to the number of multipath clusters in the environment\footnote{The \textit{Format-$2$} of the $\eta{-}\mu$ distribution is also known as $\lambda{-}\mu$ distribution.} \cite{J:Yacoub}.
 
 The SNR probability density function of the $\eta{-}\mu$ distribution is expressed as, 
\begin{equation} \label{eta-mu} 
p_{\gamma}(\gamma) = \frac{2 \sqrt{\pi}\mu^{\mu + \frac{1}{2}} h^{\mu}}{\Gamma(\mu)H^{\mu - \frac{1}{2}}} \frac{\gamma^{\mu - \frac{1}{2}}}{\overline{\gamma}^{\mu + \frac{1}{2}}} e^{-2\mu h \frac{\gamma}{\overline{\gamma}}} I_{\mu - \frac{1}{2}}\left( \frac{2\mu H \gamma}{\overline{\gamma}}\right)
\end{equation}
where  $\overline{\gamma}$ denotes the average SNR whereas
\begin{equation}  \label{eta-mu_parameters_1}
h = \frac{2 + \eta^{-1} +\eta}{4}, \, \qquad \,\, H = \frac{\eta^{-1} - \eta}{4} 
\end{equation}
in \textit{Format-$1$} with $0< \eta< \infty$ and,
\begin{equation}  \label{eta-mu_parameters_2}
h = \frac{1}{1 - \eta^{2}}, \, \qquad \,\, H = \frac{\eta}{1 - \eta^{2}}
\end{equation}
in \textit{Format-$2$} with $-1< \eta< 1$. In addition, 
\begin{equation}  \label{mu_parameter}
\mu = \frac{E^{2}(R^{2})}{2 {\rm Var}(R^{2})}\left[1 + \frac{H}{h} \right]
\end{equation} 
with $E(.)$ and $\hat{r}$ denoting expectation and the root-mean-square ${\rm (rms)}$ value of the envelope $R$, respectively \cite{J:Yacoub}. 

\subsection{Energy Detection over $\eta{-}\mu$ Fading Channels}

\begin{corollary}
For $u, \overline{\gamma}, \lambda \in \mathbb{R}^{+}$ and $\mu \in \mathbb{N}$, the following closed-form expressions hold for the average probability of detection over $\eta{-}\mu$ fading channels, 
\begin{equation} \label{eta-mu_application_1}
\begin{split}
\overline{P}_{d} &= \sum_{l = 0}^{\mu -1} \frac{(\mu)_{l} h^{\mu} \mathcal{G}\left(u, \frac{\lambda}{2} \right)}{l!   2^{\mu + l} H^{\mu + l}} \left\lbrace \frac{(-1)^{l}}{(h - H)^{\mu - l}} + \frac{(-1)^{\mu}}{(h + H)^{\mu - l}}   \right\rbrace \\
& + \sum_{l = 0}^{\mu - 1} \sum_{i = 0}^{\mu - l - 1} \frac{h^{\mu} \mu^{i} \lambda^{u} (\mu)_{l}  \overline{\gamma} e^{- \frac{\lambda}{2}} }{u! l! 2^{\mu + u + i - i} H^{\mu + l} } \times \\
& \qquad \left\lbrace \frac{(-1)^{l} \, _{1}F_{1}\left( 1 + i, 1 + u, \frac{\lambda \overline{\gamma}}{2 \overline{\gamma} + 4 \mu (h - H)} \right)}{(h - H)^{\mu - l -i} (\overline{\gamma} + 2  (h - H)\mu)^{i + 1}}  \right. \\
& \left. \quad   \qquad   + \frac{(-1)^{\mu} \, _{1}F_{1}\left( 1 + i, 1 + u, \frac{\lambda \overline{\gamma}}{2 \overline{\gamma} + 4 \mu (h + H)} \right)}{(h + H)^{\mu - l -i} (\overline{\gamma} + 2  (h + H)\mu)^{i + 1}} \right\rbrace
\end{split} 
\end{equation}
where $\mathcal{G}(a,x) \triangleq \Gamma(a,x){/}\Gamma(a)$ denotes the regularized upper incomplete gamma function and $h$ and $H$ are given by \eqref{eta-mu_parameters_1} and \eqref{eta-mu_parameters_2} according to \textit{Format}-$1$ and \textit{Format}-$2$, respectively.  
\end{corollary}

\begin{proof}
The average detection probability is obtained by averaging  \eqref{Pd_1} over the fading statistics of the channel, namely, 
\begin{equation} \label{Av_P_d}
\overline{P}_{d} = \int_{0}^{\infty} Q_{u}\left(\sqrt{2\gamma}, \sqrt{\lambda}\right) p_{\gamma}(\gamma) d \gamma. 
\end{equation}
By substituting \eqref{eta-mu} in \eqref{Av_P_d} one obtains, 
\begin{equation} \label{Av_P_d_1}
\overline{P}_{d} =\mathcal{A} \int_{0}^{\infty}  \frac{  Q_{u}\left(\sqrt{2\gamma}, \sqrt{\lambda}\right)  I_{\mu - \frac{1}{2}}\left( \frac{2\mu H \gamma}{\overline{\gamma}}\right) }{\gamma^{ \frac{1}{2} - \mu } e^{2\mu h \frac{\gamma}{\overline{\gamma}}} } d \gamma 
\end{equation}
where 
\begin{equation} \label{scalar}
\mathcal{A} = \frac{2 \sqrt{\pi}\mu^{\mu + \frac{1}{2}} h^{\mu}  }{\Gamma(\mu)H^{\mu - \frac{1}{2} }\gamma^{\mu + \frac{1}{2}}}.   
\end{equation}
Importantly, for the special case that $\mu$ is a positive integer, the Bessel function in \eqref{Av_P_d_1} can be expressed in closed-form with the aid of \cite[eq. (8.467)]{B:Ryzhik} namely,   
\begin{equation} \label{Bessel_Integer}
\begin{split}
 I_{\mu - \frac{1}{2}}\left( \frac{2\mu H \gamma}{\overline{\gamma}}\right) &= \sum_{l = 0}^{\mu - 1} \frac{  (-1)^{l}  \Gamma(\mu + l) \overline{\gamma}^{l + \frac{1}{2}}  e^{\frac{2 \mu H \gamma}{\overline{\gamma} }} }{l! \sqrt{\pi} \Gamma(\mu - l) (4 \mu H \gamma)^{l + \frac{1}{2}}}  \\
 & + \sum_{l = 0}^{\mu - 1} \frac{  (-1)^{\mu}  \Gamma(\mu + l) \overline{\gamma}^{l + \frac{1}{2}}  e^{ - \frac{2 \mu H \gamma}{\overline{\gamma} }} }{l! \sqrt{\pi} \Gamma(\mu - l) (4 \mu H \gamma)^{l + \frac{1}{2}}}.   
 \end{split}
\end{equation}
Therefore,  by substituting  \eqref{Bessel_Integer} in \eqref{Av_P_d_1} it follows that, 
\begin{equation} \label{Av_P_d_2}
\begin{split}
\overline{P}_{d} &= \sum_{l = 0}^{\mu - 1} \frac{ \mathcal{A}    (-1)^{l}  \Gamma(\mu + l) \overline{\gamma}^{l + \frac{1}{2}}  }{l! \sqrt{\pi} \Gamma(\mu - l) (4 \mu H  )^{l + \frac{1}{2}}} \int_{0}^{\infty} \frac{ Q_{u}\left(\sqrt{2\gamma}, \sqrt{\lambda} \right) d \gamma }{\gamma^{l - \mu + 1}  e^{\frac{2 \mu (h - H) \gamma}{\overline{\gamma} }} }  \\
 & + \sum_{l = 0}^{\mu - 1} \frac{ \mathcal{A}    (-1)^{\mu}  \Gamma(\mu + l) \overline{\gamma}^{l + \frac{1}{2}}  }{l! \sqrt{\pi} \Gamma(\mu - l) (4 \mu H  )^{l + \frac{1}{2}}} \int_{0}^{\infty} \frac{ Q_{u}\left(\sqrt{2\gamma}, \sqrt{\lambda} \right) d \gamma }{\gamma^{l - \mu + 1}  e^{\frac{2 \mu (h + H) \gamma}{\overline{\gamma} }} }.   
 \end{split}
\end{equation}
Notably, the integrals  in \eqref{Av_P_d_2} have the same algebraic representation as \eqref{Integral} and \eqref{second} and thus, they can be expressed in closed-form with the aid  of Theorem $1$. As a result, by performing the necessary change of variables in \eqref{T_1}, substituting in \eqref{Av_P_d_2} and carrying out long but basic algebraic manipulations, equation \eqref{eta-mu_application_1} is deduced and thus, the proof is completed.   
\end{proof}

\begin{remark}
The energy threshold in \eqref{Pf_1} can be expressed as  $ \lambda = 2 \mathcal{G}^{-1} \left(u, P_{f} \right) $, where $\mathcal{G}^{-1}(.)$ is the inverse regularized upper incomplete gamma function. To this effect \eqref{eta-mu_application_1} can be also equivalently expressed in terms of $P_{f}$ as follows:
\begin{equation} \label{eta-mu_application_2}
\begin{split}
\overline{P}_{d} &= \sum_{l = 0}^{\mu -1} \frac{(\mu)_{l} h^{\mu} P_{f}}{l!   2^{\mu + l} H^{\mu + l}} \left\lbrace \frac{(-1)^{l}}{(h - H)^{\mu - l}} + \frac{(-1)^{\mu}}{(h + H)^{\mu - l}}   \right\rbrace \\
& + \sum_{l = 0}^{\mu - 1} \sum_{i = 0}^{\mu - l - 1} \frac{\overline{\gamma}  \, h^{\mu} \mu^{i} \left[ \mathcal{G}^{-1} \left(u, P_{f} \right)\right]^{u} (\mu)_{l}   }{u! l! 2^{\mu   + i - i} H^{\mu + l} e^{ \mathcal{G}^{-1} \left(u, P_{f} \right)}} \times \\
& \qquad \left\lbrace \frac{(-1)^{l} \, _{1}F_{1}\left( 1 + i, 1 + u, \frac{ \mathcal{G}^{-1} \left(u, P_{f} \right) \overline{\gamma}}{ \overline{\gamma} + 2 \mu (h - H)} \right)}{(h - H)^{\mu - l -i} (\overline{\gamma} + 2  (h - H)\mu)^{i + 1}}  \right. \\
& \left. \quad   \qquad   + \frac{(-1)^{\mu} \, _{1}F_{1}\left( 1 + i, 1 + u, \frac{ \mathcal{G}^{-1} \left(u, P_{f} \right) \overline{\gamma}}{\overline{\gamma} + 2 \mu (h + H)} \right)}{(h + H)^{\mu - l -i} (\overline{\gamma} + 2 (h + H)\mu )^{i + 1}} \right\rbrace. 
\end{split} 
\end{equation}
\end{remark}
\noindent 
To the best of the Authors knowledge \eqref{eta-mu_application_1} and \eqref{eta-mu_application_2} have not been previously reported in the open technical literature.

\section{Numerical Results}

This section is devoted to the analysis of the behaviour of energy detection in $\eta{-}\mu$  fading conditions by means of $\overline{P}_{d}$ versus $\overline{\gamma}$ curves and complementary receiver operating characteristics (ROC) curves ($P_{m}$ versus $P_{f}$). To this end, Fig. 1 illustrates the behavior of the $\overline{P}_{d}$ versus $\overline{\gamma}$ for different values of the fading parameters $\eta$ and $\mu$ for constant time-bandwidth product $u = 3$ and the case that $P_{f} = 0.01$ and $P_{f} = 0.1$. One can notice that the average probability of detection is increased as $\eta$ increases from $0.01$ to $0.95$ for both cases of $P_{f}$. This is also the case for the $\mu$ parameter as for a fixed value of $\eta$, the $\overline{P}_{d}$ increases when $\mu  = 3$ compared to the case that $\mu = 1$. This also holds for both $P_{f} = 0.01$ and $P_{f} = 0.1$ and particularly  for moderate to high average SNR levels. 

In the same context, Fig. 2 depicts the corresponding ROC curves ($P_{m} = 1 - P_{d} $ versus $P_{f}$). The value of $P_{f}$ is assumed between $0.01$ and $0.2$ while $u=4$ and $\overline{\gamma} = 15$dB.  One can observe how the performance of the detector improves as the severity of fading is reduced in terms of both $\eta$ and $\mu$. Indicatively, for $P_{f} = 0.1$ the value of  $P_{m}$ reduces by over $70 \%$ when $\eta$ changes from $0.01$ to $0.95$ for $\mu = 1.0$ and over $65\%$ when  $\mu$ changes from $1.0$ to $2.0$ for $\eta = 0.95$. This demonstrates the sensitivity of the energy detector in multipath fading conditions and how the corresponding severity of fading can affect its performance and robustness.  
 \begin{figure}[h!]
\includegraphics[ width=9.25cm,height=6.5cm]{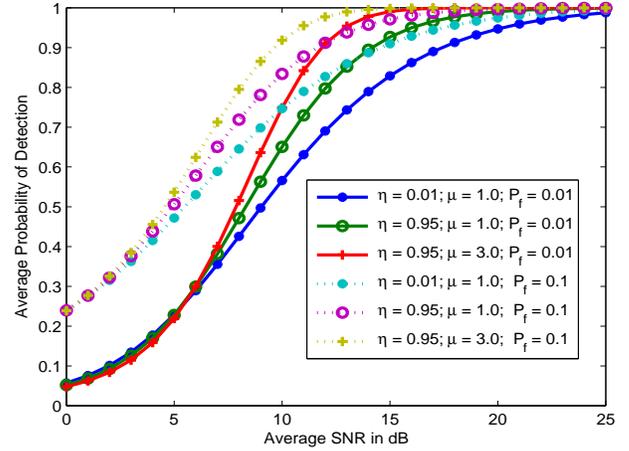} 
\caption{ $\overline{P}_{d}$ vs $\overline{\gamma}$ for different values of  $\eta$ and $\mu$ with $u = 3$ and $P_{f} = 0.01$ and $P_{f} = 0.1$. }  
\end{figure}
 \begin{figure}[h!]
\includegraphics[ width=9.25cm,height=6.5cm]{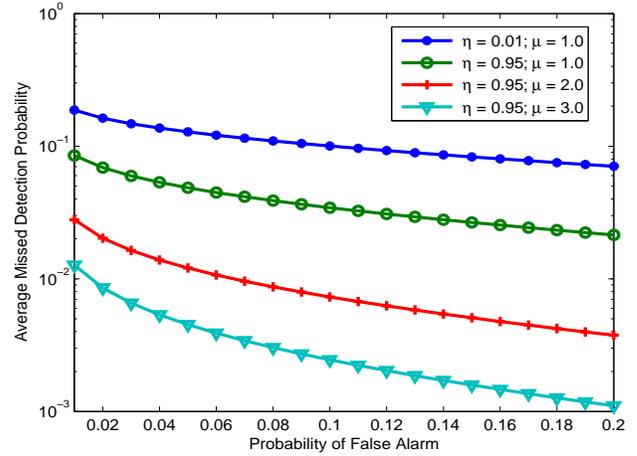} 
\caption{ $\overline{P}_{m} $ vs $P_{f}$ ROC curve for for different values of  $\eta$ and $\mu$ with $u = 4$ and $\overline{\gamma} = 15$ dB. }  
\end{figure}
\section{Conclusion}

New expressions were derived for a Marcum $Q{-}$function based integral that is often encountered in the broad area of digital communications. These expressions include an exact closed-form expression and an infinite series representation along with a closed-form  upper bound for the corresponding truncation error. The offered results have a relatively simple algebraic representation which render them analytically and numerically convenient while they can be useful in numerous applications in wireless communications.  As an example, they were used in energy detection based spectrum sensing, in the context of cognitive radio and radar systems, deriving novel closed-form expressions for the average probability of detection of unknown signals over $\eta{-}\mu$ fading channels. The derived expressions were subsequently employed in analyzing the effect of multipath fading on the spectrum sensing performance and it was shown that the overall performance of the detector is, as expected, largely affected by the value of the involved fading parameters,  particularly for moderate and high SNR levels.

%

\bibliographystyle{IEEEtran}

\end{document}